\documentclass[12pt,a4paper]{iopart}

\pdfoutput=1

\expandafter\let\csname equation*\endcsname\relax
\expandafter\let\csname endequation*\endcsname\relax

\usepackage{amsmath,amssymb}
\usepackage{enumitem}

\newtheorem{theorem}{Theorem}
\newtheorem{proposition}[theorem]{Proposition}
\newtheorem{corollary}[theorem]{Corollary}
\newenvironment{proof}{\textbf{Proof:}}{\rule{0pt}{0pt}\hfill$\Box$}
\newenvironment{proofsketch}{\textbf{Proof (sketch):}}{\rule{0pt}{0pt}\hfill$\Box$\par}

\usepackage{orcidlink}

\def\C{\mathbb{C}}
\def\F{\mathbb{F}}

\def\wgt{\mathop{\mathrm{wgt}}}

\begin{document}

\title{New Quantum Codes from CSS Codes}

\author{Markus Grassl${}^{1,2\,\orcidlink{0000-0002-3720-5195}}$}

\address{${}^1$ International Centre for Theory of Quantum
  Technologies (ICTQT),\\\phantom{${}^2$} University of Gdansk, 80-309 Gda\'nsk, Poland}
\address{${}^2$ Max Planck Institute for the Science of Light, 91058 Erlangen, Germany}
  
\ead{markus.grassl@ug.edu.pl}
\vspace{10pt}
\begin{indented}
\item[]December 2022
\end{indented}

\begin{abstract}
We present a new propagation rule for CSS codes. Starting with a CSS
code $[\![n,k,d]\!]_q$, we construct a CSS code with parameters
$[\![n-2,k,d-1]\!]_q$. In general, one would only obtain a code with
parameters $[\![n-2,k,d-2]\!]_q$. The construction applies to
asymmetric quantum codes from the CSS construction as well.
\end{abstract}

\vspace{2pc}

\section{Introduction}
Quantum error-correcting codes are an essential ingredient for the
realization of universal fault-tolerant quantum computers.  It turns
out that, compared to general stabilizer codes, the class of so-called
CSS codes \cite{CaSh96,Ste96} allows more simple, direct
implementations of a larger class of fault-tolerant operations.  For
example, any CSS code allows the transversal implementation of the
CNOT gate \cite{Got98}.

While the error-correcting capabilities of CSS codes are in general
less than those of the larger class of stabilizer codes \cite{CEL99},
there are also a couple of advantages in terms of
error-correction. The complexity of error correction is exponential in
the number of independent syndromes \cite{HsLG11}. For CSS codes, the
syndromes can be grouped into syndromes for $X$- and $Z$-errors, and
error correction can be performed independently for the two types of
errors. This yields a square-root speedup.  Likewise, in order to
determine the minimum distance of a CSS code, one has to compute that
of two codes over an alphabet of size $q$, while for stabilizer codes,
the alphabet size is $q^2$.  There is again a square-root advantage.

After a short summary of the required basic notions of both classical
and quantum codes in the next section, we present our main results in
Section \ref{sec:main}.  Then we apply the new construction to CSS codes
$[\![n,0,d]\!]_2$ for qubits, yielding explicit constructions of codes
with the best parameters known to date.

\section{Background}
In this section, we recall basic definitions and properties of both
classical and quantum error-correcting codes. For more details, see,
e.g., the book \cite{MS77} for classical codes and the articles
\cite{KKKS06,Gra21} for quantum codes.
\subsection{Classical codes}
Let $q=p^m$ be a prime power. Then there is a finite field $\F_q$ with
$q$ elements. By $C=[n,k,d]_q$ we denote a \emph{linear block} code of
length $n$, dimension $k$, and minimum distance (at least) $d$. The
\emph{minimum (Hamming) distance} is the minimum number of positions
in which two codewords of $C$ differ.  For linear codes, the minimum
distance equals the \emph{minimum (Hamming) weight} of the code, i.e.,
the minimum number of positions at which a non-zero codeword of $C$ is
non-zero.

The (Euclidean) \emph{dual code} $C^\perp=[n,n-k,d^\perp]_q$ is defined as
\begin{alignat}{5}
  C^\bot=\{\mathbf{v}\in\F_q\colon
  \mathbf{v}\cdot\mathbf{c}=\sum_{i=1}^n c_i v_i =0\text{ for
    all $\mathbf{c}\in C$}\}.
\end{alignat}

For $1\le i\le n$, the \emph{shortened code} $\sigma_i(C)$ is obtained
by selecting all codewords $\mathbf{c}\in C$ that are zero at the
$i$-th position, and then deleting that position, i.e.,
\begin{alignat}{5}
  \sigma_i(C) =\{ (c_1,\ldots,c_{i-1},c_{i+1},\ldots,c_n)\in\F_q^{n-1}\colon
  (c_1,\ldots,c_{i-1},0,c_{i+1},\ldots,c_n)\in C\}.\label{eq:shortening}
\end{alignat}
As we are taking a subset of the codewords of the original code, the
minimum distance does not decrease.  Similarly, the \emph{punctured
  code} $\pi_i(C)$ is obtained by deleting the $i$-th coordinate of
$C$, i.e.,
\begin{alignat}{5}
  \pi_i(C)=\{ (c_1,\ldots,c_{i-1},c_{i+1},\ldots,c_n)\in\F_q^{n-1}\colon (c_1,\ldots,c_{i-1},c_i,c_{i+1},\ldots,c_n)\in C\}.\label{eq:puncturing}
\end{alignat}
In this case, the minimum distance is in general decrease by one, as
we are removing one coordinate.  When we omit the index $i$,
$\sigma(C)$ and $\pi(C)$ refer to the shortened/punctured code with
respect to the last position $i=n$.

The following proposition summarizes properties of the shortened and
punctured code (see, e.g., \cite[Sections 1.5.1 and 1.5.3]{HuPl10}).
\begin{proposition}\label{prop:shorten_puncture}
  For a linear code $C=[n,k,d]_q$ with $d>1$ and coordinate $i$ not
  being identically zero, the shortened and punctured code have
  parameters
  \begin{alignat}{5}
     \sigma_i(C)&{}=[n-1,k-1,d]_q\label{eq:para_shorten}\\
 \text{and}\qquad \pi_i(C)&{}=[n-1,k,d-1]_q.\label{eq:para_puncture}
  \end{alignat}
  Moreover, shortening the dual code equals the dual of the punctured
  code and vice versa:
  \begin{alignat}{5}
    \sigma_i(C^\perp)&{}=\left(\pi_i(C)\right)^\bot\\
\text{and}\qquad \pi_i(C^\perp)&{}=\left(\sigma_i(C)\right)^\bot.
  \end{alignat}
\end{proposition}
When coordinate $i$ of the code $C$ is identically zero, the dimension
of the shortened code remains $k$, as \eqref{eq:shortening} involves
all codewords of $C$.  Likewise, for $d=1$, when
puncturing the code at the position $i$ where a codeword of weight $1$
is non-zero, the dimension drops to $k-1$, as deleting that position
results in the zero vector.  Moreover, in some cases the
minimum distance of the shortened or punctured code might be larger,
e.g., when all minimum weight codewords are non-zero at that
position (shortening) or all minimum weight words are zero at that
position (puncturing).

Later, we will use so-called quadratic residue (QR) codes. Their
properties are summarized as follows (for the precise definitions,
see, e.g., \cite[Chapter 16]{MS77} for prime fields, and \cite[Section
  6.6]{HuPl10} or \cite[Section 4.6]{BBFKKW06} for the general
case).
\begin{proposition}[QR codes]\label{prop:QRcodes}
  Let $q$ be a prime power and $n>2$ a prime such that $q$ is a
  quadratic residue modulo $n$ (i.e., $q=a^2\bmod n$ for some
  integer $a$). The quadratic residue codes $\mathcal{Q}$,
  $\overline{\mathcal{Q}}$, $\mathcal{N}$, and
  $\overline{\mathcal{N}}$ are cyclic codes of length $n$ over $\F_q$
  with the following properties:
  \begin{enumerate}[leftmargin=*]
    \item $\overline{\mathcal{Q}} < \mathcal{Q}$ and  $\overline{\mathcal{N}} < \mathcal{N}$
    \item $\mathcal{Q}$ and $\mathcal{N}$ are equivalent with
      $\mathcal{Q}=[n,(n+1)/2,d]_q$ and $\mathcal{N}=[n,(n+1)/2,d]_q$.
    \item $\overline{\mathcal{Q}}$ and $\overline{\mathcal{N}}$ are
      equivalent with $\overline{\mathcal{Q}}=[n,(n-1)/2,d+1]_q$ and
      $\overline{\mathcal{N}}=[n,(n-1)/2,d+1]_q$.
    \item $\mathcal{Q}^\perp = \overline{\mathcal{Q}}$ and
      $\mathcal{N}^\perp=\overline{\mathcal{N}}$ for $n\equiv 3\mod 4$
    \item $\mathcal{Q}^\perp = \overline{\mathcal{N}}$ and
      $\mathcal{N}^\perp=\overline{\mathcal{Q}}$ for $n\equiv 1\mod 4$
  \end{enumerate}
  The \emph{extended QR codes} $\widehat{\mathcal{Q}}$ and
  $\widehat{\mathcal{N}}$ are obtained adding one coordinate to $\mathcal{Q}$ and
  $\mathcal{N}$, respectively, such that
  $\widehat{\mathcal{Q}}^\perp=\widehat{\mathcal{N}}$ for $n\equiv
  1\mod 4$, while $\widehat{\mathcal{Q}}^\perp=\widehat{\mathcal{Q}}$
  and $\widehat{\mathcal{N}}^\perp=\widehat{\mathcal{N}}$ for $n\equiv
  3\mod 4$. The extended QR codes have the following properties:
  \begin{enumerate}[leftmargin=*]
    \item[(vi)] $\widehat{\mathcal{Q}}=[n+1,(n+1)/2,d+1]_q$ and
      $\widehat{\mathcal{N}}=[n+1,(n+1)/2,d+1]_q$
    \item[(vii)] The automorphism group of the extended QR code
      $\widehat{\mathcal{Q}}$ contains the projective special linear
      group $PSL_2(n)$.
  \end{enumerate}
  
\end{proposition}

\subsection{Quantum codes}
By $\mathcal{C}=[\![n,k,d]\!]_q$ we denote a quantum-error correcting
code (QECC) that is a subspace of dimension $q^k$ of the complex
vector space $\mathcal{H}=\left(\C^q\right)^{\otimes
  n}\cong\C^{q^n}$. When the minimum distance is $d$, all errors
acting on at most $d-1$ of the $n$ subsystems (tensor components) can
be detected or act trivially on the code. Alternatively, one can
correct up to $d-1$ erasures or up to $\lfloor (d-1)/2\rfloor$
errors. Most of the QECCs in the literature are so-called stabilizer
codes which can be constructed from certain additive codes over the
field $\F_{q^2}$ \cite{KKKS06}.  Here we focus on so-called CSS codes,
named after Calderbank \& Shor \cite{CaSh96} and Steane \cite{Ste96}.

\begin{proposition}[CSS construction]
  Let $C_1=[n,k_1,d_1']_q$ and $C_2=[n,k_2,d_2']_q$ be linear codes over
  $\mathbb{F}_q$ with $C_2^\bot\le C_1$. Then there is a CSS code with
  parameters $\mathcal{C}=[\![n,k_1+k_2-n,d]\!]_q$. The minimum
  distance is $d = \min(d_1,d_2)\ge \min(d_1',d_2')$. For $C_2^\perp\ne C_1$, 
  \begin{alignat}{5}
    d_1&{}=&\wgt(C_1\setminus C_2^\perp)&{}=\min\{\wgt({\bf v})\colon{\bf v}\in C_1\setminus C_2^\perp\}&&{}\ge d_1'\\
    \text{and}\qquad
    d_2&{}=&\wgt(C_2\setminus C_1^\perp)&{}=\min\{\wgt({\bf v})\colon{\bf v}\in C_2\setminus  C_1^\perp\}&&{}\ge d_2'.
  \end{alignat}
  For $C_2^\perp=C_1$, $d_1=d_1'$ and $d_2=d_2'$. Then the CSS code has parameters
  $\mathcal{C}=[\![n,0,d]\!]_q$ with $d=\min(d_1',d_2')$.
\end{proposition}
As CSS codes are a subclass of stabilizer codes, their parameters need
not be optimal. On the other hand, CSS codes allow to correct $X$- and
$Z$-errors independently.  Steane \cite{Ste96} uses the notation
$\{n,k,d_1,d_2\}$ in order to stress that for $d_1 > d_2$, the code
can correct more, say, $Z$-errors than $X$-errors.  Such codes are
known as \emph{asymmetric quantum codes} \cite{IoMe07}, and the
notation $[\![n,k,d_1/d_2]\!]_q$ is used in the literature
\cite{SKR09}.  To avoid possible confusion, we will use the notation
$[\![n,k,\{d_1,d_2\}]\!]_q$, which indicates both minimum distances
and also stresses that the role of the codes $C_1$ and $C_2$ in the
CSS construction can be swapped.

\section{Main results}\label{sec:main}
For a quantum code $\mathcal{C}=[\![n,k,d]\!]_q$ that is not the
tensor product of quantum codes of lengths $n-1$ and $1$, deleting any
coordinate results in a code $\mathcal{C}=[\![n-1,k,d-1]\!]_q$ (see
\cite[Theorem 6 d)]{CRSS98}). If the code is
$\mathcal{C}=[\![n,k,d]\!]_q$ a tensor product
$\mathcal{C}=\mathcal{C}_{n-1}\otimes\mathcal{C}_1$ of codes of lengths
$n-1$ and $1$, respectively, then deleting the last coordinate yields
the code $\mathcal{C}_{n-1}$ which has parameters
$\mathcal{C}_{n-1}=[\![n-1,k,d]\!]_q$ for
$\mathcal{C}_1=[\![1,0,1]\!]_q$, or a code
$\mathcal{C}_{n-1}=[\![n-1,k-1,d']\!]_q$ for
$\mathcal{C}_1=[\![1,1,1]\!]_q$. In the latter case, the original code
has distance $d=1$, and $d'$ is the distance of the tensor factor
$\mathcal{C}_{n-1}$. 

In the following, we show that for CSS codes, we can reduce the length
by $2$ while the minimum distance is decreased only by $1$.

\begin{theorem}\label{thm:puncture_CSS}
  Let $\mathcal{C}=[\![n,k,\{d_1,d_2\}]\!]_q$ with $d_1>1$ be a CSS code.
  Then for all $i\in\{1,\ldots,n\}$, there is a CSS code
  $\mathcal{C}_i=[\![n-1,k,\{d_1-1,d_2\}]\!]_q$.
\end{theorem}
\begin{proof}
  Let $C_1=[n,k_1,d_1']_q$ and $C_2=[n,k_2,d_2']_q$ be the classical
  codes used for the construction of the CSS code
  $\mathcal{C}=[\![n,k,\{d_1,d_2\}]\!]_q$. Then $C_2^\perp\le C_1$ and
  $k=k_1+k_2-n$. For a fixed position $i\in\{1,\ldots,n\}$, consider
  the codes $\pi_i(C_1)$ and $\sigma_i(C_2)$ of length $n-1$. By
  Proposition \ref{prop:shorten_puncture},
  $\left(\sigma_i(C_2)\right)^\perp=\pi_i(C_2^\perp)$, and hence
  $\left(\sigma_i(C_2)\right)^\perp\le \pi_i(C_1)$. Using the codes
  $\sigma_i(C_2)$ and $\pi_i(C_1)$ in the CSS construction, we obtain
  a code
  $\widetilde{\mathcal{C}}=[\![n-1,\widetilde{k},\{\widetilde{d}_1,\widetilde{d}_2\}]\!]_q$.

  It remains to show that the code has the claimed
  parameters. Concerning the dimension $k$, we distinguish two cases.
  \begin{enumerate}
    \item When the coordinate $i$ in the code $C_2$ is identically
      zero, then the shortened code has parameters
      $\sigma_i(C_2)=[n-1,k_2,d_2']_q$. Moreover, the dual code
      $C_2^\perp\le C_1$ contains a word of weight one that is
      non-zero at position $i$. Then the punctured code has parameters
      $\pi_i(C_1)=[n-1,k_1-1,\widetilde{d}_1']_q$. The dimension of
      the CSS code is $\widetilde{k}=(k_1-1)+k_2-(n-1)=k_1+k_2-n=k$.
  \item When the coordinate $i$ in the code $C_2$ is not identically
    zero, the shortened code has parameters
    $\sigma_i(C_2)=[n-1,k_2-1,d_2']_q$. Moreover, the dual code
    $C_2^\perp\le C_1$ does not contain a word of Hamming weight one
    that is non-zero only at position $i$. As $d_1>1$, $C_1$ does not
    contain such a word of weight one either. Then the punctured code
    has parameters $\pi_i(C_1)=[n-1,k_1,d_1'-1]_q$. Again
    $\widetilde{k}=k_1+(k_2-1)-(n-1)=k_1+k_2-n=k$.
  \end{enumerate}
  Concerning the minimum distances $\widetilde{d}_1$ and
  $\widetilde{d}_1$, first consider the case $k=0$.  In this case, the
  minimum distances of asymmetric CSS codes are equal those of the
  classical codes, and hence $\widetilde{d}_1=d_1-1$ and
  $\widetilde{d}_2=d_2$ as claimed. For $k>0$, we distinguish two
  cases. To simplify notation, we set $i=n$.
  \begin{enumerate}
  \item Assume that
    $\pi_i(C_1)\setminus\left(\sigma_i(C_2)\right)^\perp$ contains a
    word $\widetilde{\mathbf{v}}=(v_1,\ldots,v_{n-1})$ of weight
    strictly less than $d_1-1$. Then the code $C_1$ contains a word
    $\mathbf{v}=(v_1,\ldots,v_{n-1},v_n)$ of weight strictly less than
    $d_1=\wgt(C_1\setminus C_2^\perp)$. Therefore $\mathbf{v}\in
    C_2^\perp$, which implies
    $\widetilde{\mathbf{v}}\in\pi_i(C_2^\perp)=\left(\sigma_i(C_2)\right)^\perp$,
    a contradiction.  Hence $\widetilde{d}_1\ge d_1-1$.
  \item Assume that
    $\sigma_i(C_2)\setminus\left(\pi_i(C_1)\right)^\perp$ contains a
    word $\widetilde{\mathbf{u}}=(u_1,\ldots,u_{n-1})$ of weight
    strictly less than $d_2$. Then the code $C_2$ contains a word
    $\mathbf{u}=(u_1,\ldots,u_{n-1},0)$ of weight strictly less than
    $d_2=\wgt(C_2\setminus C_1^\perp)$. Therefore $\mathbf{u}\in
    C_1^\perp$, which implies
    $\widetilde{\mathbf{u}}\in\sigma_i(C_1^\perp)=\left(\pi_i(C_1)\right)^\perp$,
    a contradiction.  Hence $\widetilde{d}_2\ge d_2$.
  \end{enumerate}
\end{proof}

Note that this result can also be found in \cite[Theorem 4.2 (i)]{LaGu13},
using a slightly different proof technique.  La Guardia also discusses
the special case \cite[Theorem 4.2 (ii)]{LaGu13} that puncturing the
classical code $C_1$ does not decrease its minimum distance, and hence
the minimum distances of the resulting asymmetric quantum code do not
change. The following direct consequence of
Theorem~\ref{thm:puncture_CSS} appears to be new:

\begin{theorem}\label{thm:main}
  Let $\mathcal{C}=[\![n,k,\{d_1,d_2\}]\!]_q$ with $d_1,d_2>1$ be a
  CSS code. Then there is a CSS code
  $\widehat{\mathcal{C}}=[\![n-2,k,\{d_1-1,d_2-1\}]\!]_q$.
\end{theorem}
\begin{proof}
  Using Theorem \ref{thm:puncture_CSS}, we obtain a CSS code
  $\mathcal{C}'=[\![n-1,k,\{d_1-1,d_2\}]\!]_q$. Switching the roles of
  $C_1$ and $C_2$ in the CSS construction, again using Theorem
  \ref{thm:puncture_CSS} we obtain a CSS code
  $\widehat{\mathcal{C}}=[\![n-2,k,\{d_1-1,d_2-1\}]\!]_q$.
\end{proof}

\begin{corollary}
  Let $\mathcal{C}=[\![n,k,d]\!]_q$ with $d>1$ be a
  CSS code. Then there is a CSS code
  $\widehat{\mathcal{C}}=[\![n-2,k,d-1]\!]_q$.
\end{corollary}
\begin{proof}
  This follows directly from Theorem \ref{thm:main} together with the
  fact that the minimum distance of a CSS code
  $\widehat{\mathcal{C}}=[\![n-2,k,\{d_1-1,d_2-1\}]\!]_q$ is
  $d=\min(d_1,d_2)$.
\end{proof}

In the literature, usually the general situation of puncturing quantum
codes is discussed (see, e.g., \cite[Theorem 6 d)]{CRSS98} for qubit
codes and \cite[Proposition 2.8]{FLX06} for $q$-dimensional quantum
systems), i.e., the existence of a code $[\![n,k,d]\!]_q$ with $n,d\ge
2$ implies the existence of a code $[\![n-1,k,d-1]\!]_q$. Puncturing a
code twice hence results in a code with parameters
$[\![n-2,k,d-2]\!]_q$, while for CSS codes the minimum distance only
drops by one.

Triggered by a question by Chinmay Nirkhe, we note that the results of
Theorems \ref{thm:puncture_CSS} and \ref{thm:main} are the best
possible for general codes.  For asymmetric CSS codes
$[\![n,k,\{d_1,d_2\}]\!]_q$, the Singleton bound is $n+2\ge k+d_1+d_2$
\cite[Lemma 3.3]{SKR09}, while for codes $[\![n,k,d]\!]_q$ the
Singleton bound reads $n+2\ge k+2d$ \cite{KnLa97}. When starting with
a quantum MDS CSS code, i.e., a code that meets the bound with
equality, reducing the length by one and preserving the dimension
implies that one of the minimum distance has to be reduced by one as
well.  Similarly, starting with a QMDS codes $[\![n,k,d]\!]_q$, when
reducing the length by two the distance has to be reduced by at last
one.

\section{Examples}
First, we apply our results to quantum quadratic residue codes (for
details, see \cite[Section IX]{KKKS06} (for $k=1$) and
\cite{DaLi22,DaLi22b}).
\begin{proposition}[quantum QR codes]\label{prop:quantum_QR_codes}
  Let $q$ be a prime power and $n>2$ a prime such that $q$ is a
  quadratic residue modulo $n$. Then there are CSS codes with
  parameters
  \begin{alignat}{5}
     && \mathcal{C}_{\text{QR}}&{}=[\![n,1,d]\!]_q\\
   \text{and}\qquad && \widehat{\mathcal{C}}_{\text{QR}}&{}=[\![n+1,0,d+1]\!]_q,
  \end{alignat}
  where $d$ is the minimum distance of the classical QR code $\mathcal{Q}=[n,(n+1)/2,d]_q$.
\end{proposition}
\begin{proofsketch}
  For the code $\mathcal{C}_{\text{QR}}=[\![n,1,d]\!]_q$, in the case
  $n\equiv 3\mod 4$, we use $C_1=C_2=\mathcal{Q}$. By Proposition
  \ref{prop:QRcodes} (i) and (iv),
  $C_2^\perp=\mathcal{Q}^\perp=\overline{\mathcal{Q}}<\mathcal{Q}=C_1$. When
  $n\equiv 1\mod 4$, we use $C_1=\mathcal{Q}$ and $C_2=\mathcal{N}$.
  By Proposition \ref{prop:QRcodes} (i) and (v),
  $C_2^\perp=\mathcal{N}^\perp=\overline{\mathcal{Q}}<\mathcal{Q}=C_1$.
  As the dimensions of $\mathcal{Q}$ and $\overline{\mathcal{Q}}$
  differ by one, we get a CSS code of dimension $k=1$.  Both
  $\mathcal{Q}$ and $\mathcal{N}$ have parameters $[n,(n+1)/2,d]_q$,
  and hence the minimum distance of the CSS code equals that of the
  classical QR codes.

  For the code
  $\widehat{\mathcal{C}}_{\text{QR}}=[\![n+1,0,d+1]\!]_q$, we use
  $C_1=C_2^\perp=\widehat{\mathcal{Q}}$. Then $C_1^\perp=C_2$ equals
  either $\widehat{\mathcal{Q}}$ for $n\equiv 3\mod 4$, or
  $\widehat{\mathcal{N}}$ for $n\equiv 3\mod 1$.  Both
  $\widehat{\mathcal{Q}}$ and $\widehat{\mathcal{N}}$ have parameters
  $[n+1,(n+1)/2,d+1]_q$. Hence then minimum distance of the CSS code
  equals $d+1$.
\end{proofsketch}
Applying Theorem \ref{thm:main} repeatedly to quantum QR codes
$\widehat{\mathcal{C}}_{\text{QR}}$ with $q=2$, we obtain qubit CSS
codes with the parameters shown in Table \ref{tab:CSScodes1}.  The
parameters of the codes in the first column of the table can also be
found in \cite{DaLi22}, together with references for particular codes.
All codes in Table \ref{tab:CSScodes1} have the largest minimum
distance among the explicitly known stabilizer codes (see
\cite{Grassl:codetables}). As quantum codes with $k=0$ are by
definition pure, a code with parameters $[\![n,0,d]\!]_q$ implies the
existence of codes $[\![n-k,k,d-k]\!]_q$ for $0\le k< d$ (see
\cite[Lemma 70]{KKKS06}). For the codes of Table \ref{tab:CSScodes1},
however, the minimum distance of the resulting codes is smaller than
the best known stabilizer codes in \cite{Grassl:codetables}.  The same
is true for quantum QR codes with $k=1$.  The derived CSS codes with
$k>0$ might nonetheless find application in a fault-tolerant context.

\begin{table}[hbt]
  \caption{Parameters of good CSS codes from extended QR codes. When
    the minimum distance is larger than the lower bound by Theorem
    \ref{thm:main}, it is marked with ${}^*$.\label{tab:CSScodes1}}
  \begin{footnotesize}
    \def\arraystretch{1.2}
    \tabcolsep0.8\tabcolsep
    \medskip
  \begin{tabular}{|l|llllll|}
    \hline
    $\widehat{\mathcal{C}}_{\text{QR}}$ &
    \multicolumn{6}{l|}{codes from Theorem \ref{thm:main}}\\
    \hline
    &&&&&&\\[-4ex]
    \hline
    $[\!  [8, 0, 4]\!]_2$  &&&&&& \\
    \hline
    $[\! [24, 0, 8]\!]_2$  &&&&&& \\
    \hline
    $[\![138, 0, 22]\!]_2$ &&&&&& \\
    \hline
    $[\![168, 0, 24]\!]_2$ & $[\![166, 0, 23]\!]_2$ & $[\![164, 0, 22]\!]_2$ & $[\![162, 0, 21]\!]_2$&&&\\
    \hline
    $[\![192, 0, 28]\!]_2$ & $[\![190, 0, 27]\!]_2$ & $[\![188, 0, 26]\!]_2$ & $[\![186, 0, 25]\!]_2$ 
                           & $[\![184, 0, 24]\!]_2$ & $[\![182, 0, 23]\!]_2$ & $[\![180, 0, 22]\!]_2$ \\
    \hline
    $[\![200, 0, 32]\!]_2$ & $[\![198, 0, 31]\!]_2$ & $[\![196, 0, 30]\!]_2$ & $[\![194, 0, 29]\!]_2$&&& \\
    \hline
    $[\![224, 0, 32]\!]_2$ & $[\![222, 0, 31]\!]_2$ & $[\![220, 0, 30]\!]_2$ & $[\![218, 0, 29]\!]_2$
                           & $[\![216, 0, 28]\!]_2$ & $[\![214, 0, 27]\!]_2$ & $[\![212, 0, 26]\!]_2$ \\
                           & $[\![210, 0, 25]\!]_2$ & $[\![208, 0, 25^*]\!]_2$ & $[\![206, 0, 24^*]\!]_2$
                           & $[\![204, 0, 24^*]\!]_2$ & $[\![202, 0, 24^*]\!]_2$ &\\
    \hline
  \end{tabular}
  \end{footnotesize}
\end{table}

\begin{table}[hbt]
  \caption{Parameters of good CSS codes from self-dual binary codes.\label{tab:CSScodes2}}
  \begin{footnotesize}
    \def\arraystretch{1.2}
    \medskip
  \begin{tabular}{|l|llllll|}
    \hline
    $\mathcal{C}$ &
    \multicolumn{6}{l|}{codes from Theorem \ref{thm:main}}\\
    \hline
    &&&&&&\\[-4ex]
    \hline
    $[\![136, 0, 24]\!]_2$ & $[\![134, 0, 23]\!]_2$ & $[\![132, 0, 22]\!]_2$ & $[\![130, 0, 21]\!]_2$&&&\\
    \hline
    $[\![152, 0, 24]\!]_2$ & $[\![150, 0, 23]\!]_2$ & $[\![148, 0, 22]\!]_2$ & $[\![146, 0, 21]\!]_2$
                           & $[\![144, 0, 20]\!]_2$ & $[\![142, 0, 19]\!]_2$ & $[\![140, 0, 18]\!]_2$ \\
    \hline
    $[\![160, 0, 24]\!]_2$ & $[\![158, 0, 23]\!]_2$ & $[\![156, 0, 22]\!]_2$ & $[\![154, 0, 21]\!]_2$ &&&\\
    \hline
  \end{tabular}
  \end{footnotesize}
\end{table}

Note that the binary extended QR code of length $152$ has only minimum
distance $20$, while there is a self-dual binary code $[152,76,24]_2$
(see \cite{Grassl:codetables}).  This yields a CSS code
$[\![152,0,24]\!]_2$.  Moreover, there are self-dual binary codes
$[136,68,24]_2$ and $[160,80,24]_2$ which yield CSS codes
$[\![134,0,24]\!]_2$ and $[\![160,0,24]\!]_2$, respectively.  Applying
Theorem \ref{thm:main} to these codes yields the codes shown in Table
\ref{tab:CSScodes2}.

For the codes $[\![206,0,24^*]\!]_2$, $[\![204,0,24^*]\!]_2$, and
$[\![202,0,24^*]\!]_2$ derived from the code
$\widehat{\mathcal{C}}_{\text{QR}}=[\![224,0,32]\!]_2$, the true
minimum distance is larger than what is guaranteed by Theorem
\ref{thm:main}.  This is related to the fact that when puncturing a
classical code at $m$ positions, the minimum distance will be larger
than $d-m$ in case every codeword of minimum weight has at least one
zero at the $m$ positions \cite{GrWh04}. For the code
$[\![202,0,24^*]\!]_2$, by Theorem \ref{thm:main}, the minimum
distance is at least $21$.  Using Magma V2.27-3 \cite{Magma}, we have
verified that the minimum distance is $24$.  The calculation took
about $20$ hours wall-clock time using $48$ cores.  When we apply
Theorem \ref{thm:main} repeatedly, then the minimum distance of the
codes in the sequence is monotonically decreasing.  This implies that
the codes of length $206$ and $204$ have minimum distance at least
$24$ as well, while the lower bounds are $23$ and $22$, respectively.

In general, to compute the minimum distance of a CSS code
$\mathcal{C}=[\![n,0,\{d_1,d_2\}]\!]_q$ derived from classical codes
$C_1=[n,k,d_1]_q$ and $C_2=[n,n-k,d_2]_q$ with $C_2^\perp=C_1$, one
would have to compute the minimum distance of both $C_1$ and $C_2$. By
Proposition \ref{prop:QRcodes} (vii), the automorphism group of an
extended QR code $\widehat{\mathcal{Q}}=[n+1,(n+1)/2,d+1]_q$ contains
the group $PSL_2(n)$. In turn $PSL_2(n)$ contains an involution $\tau$
with zero ($n\equiv 3\bmod 4$) or two ($n\equiv 1 \bmod 4$) fixed
points. We can select the positions $I_s$ for shortening and the
positions $I_p$ for puncturing such that $I_s^\tau = I_p$. In this
case, the shortened and punctured codes $C_1$ and $C_2$ will be
permutation equivalent, i.e., they have the same parameters. Hence we
have to compute the minimum distance of only one of the codes.

For the code  $[\![208,0,25^*]\!]_2$, the lower bound on the minimum
distance by Theorem \ref{thm:main} is $24$, but it turns out that the
true minimum distance is $25$. Using $48$ cores, the verification
took about $125$ hours wall-clock time.

The minimum distance of the codes in Tables \ref{tab:CSScodes1} and
\ref{tab:CSScodes2} is the largest among the codes that have been
explictly constructed \cite{Grassl:codetables}.  At the same time, for
$n>100$, the Gilbert-Varshamov bound in \cite[Corollary 2.3,
  1)]{FeMa04} implies the \emph{existence} of quantum codes
$[\![n,0,d]\!]_q$ for even $n$ with larger minimum distance.  However,
for length $n>100$ it is essentially infeasible to determine the
minimum distance of a random quantum code.

Further examples can be derived from the extended ternary QR code
$C=[60,30,18]_3$, which yields a CSS code $[\![60,0,18]\!]_3$. From
this code, we can derive the codes in Table \ref{tab:CSScodes3} which
achieve the Gilbert-Varshamov (GV) bound in \cite[Corollary
  2.3]{FeMa04}. The minimum distance of the codes
$[\![59,1,17^*]\!]_3$ and $[\![60,0,18^*]\!]_3$ is larger than the GV
bound.  The codes in Table \ref{tab:CSScodes3} have the largest
minimum distance among codes known to the author.

\begin{table}
  \caption{Parameters of good ternary CSS codes from the self-dual
    ternary extended QR code $[60,30,18]_3$.\label{tab:CSScodes3}}
  \medskip
  
    \centerline{
    \begin{footnotesize}
    \def\arraystretch{1.2}
    \tabcolsep0.8\tabcolsep
    \medskip
    \begin{tabular}{|l|l|l|l|l|l|}
      \hline
      $[\![56,0,16]\!]_3$ & &\\
      \hline
      $[\![57,0,16]\!]_3$ & $[\![57,1,16]\!]_3$ & $[\![57,2,15]\!]_3$ \\
      \hline
      $[\![58,0,17]\!]_3$ & $[\![58,1,16]\!]_3$ & $[\![58,2,16]\!]_3$ \\
      \hline
      $[\![59,0,17]\!]_3$ & $[\![59,1,17^*]\!]_3$ & \\
      \hline
      $[\![60,0,18^*]\!]_3$ & &\\
      \hline
    \end{tabular}
    \end{footnotesize}}
\end{table}

\section{Conclusions}
We have presented a propagation rule for the parameters of CSS codes
that results in better parameters than the corresponding construction
for general quantum error-correcting codes. In the case of qubit codes
$[\![n,0,d]\!]_2$, this allows us to explicitly construct CSS codes
with the largest minimum distance among the known codes.  As already
mentioned out in last section, counting arguments imply
the existence of possibly non-CSS codes with even larger minimum
distance, but these have yet to be constructed.

\section*{Acknowledgments}
The author would like to thank Petr Lison\v{e}k for discussions during
the workshop WCC 2022 that lead to the results presented in this
work.  He also thanks the referees for pointing out the connection to
the results by La Guardia \cite{LaGu13}.

The `International Centre for Theory of Quantum Technologies’ project
(contract no. MAB/2018/5) is carried out within the International
Research Agendas Programme of the Foundation for Polish Science
co-financed by the European Union from the funds of the Smart Growth
Operational Programme, axis IV: Increasing the research potential
(Measure 4.3).


\section*{References}     

\begin{thebibliography}{10}

\bibitem{CaSh96}
A.~Robert Calderbank and Peter~W. Shor.
\newblock Good quantum error-correcting codes exist.
\newblock {\em Physical Review~A}, 54(2):1098--1105, August 1996.
\newblock \href {https://doi.org/10.1103/PhysRevA.54.1098}
  {\path{doi:10.1103/PhysRevA.54.1098}}.

\bibitem{Ste96}
Andrew~M. Steane.
\newblock Simple quantum error correcting codes.
\newblock {\em Physical Review A}, 54(6):4741--4751, December 1996.
\newblock \href {https://doi.org/10.1103/PhysRevA.54.4741}
  {\path{doi:10.1103/PhysRevA.54.4741}}.

\bibitem{Got98}
Daniel Gottesman.
\newblock Theory of fault-tolerant quantum computation.
\newblock {\em Physical Review A}, 57(1):127--137, January 1998.
\newblock \href {https://doi.org/10.1103/PhysRevA.57.127}
  {\path{doi:10.1103/PhysRevA.57.127}}.

\bibitem{CEL99}
G{\'e}rard Cohen, Sylvia Encheva, and Simon Litsyn.
\newblock On binary constructions of quantum codes.
\newblock {\em IEEE Transactions on Information Theory}, 45(7):2495--2498,
  November 1999.
\newblock \href {https://doi.org/10.1109/18.796389}
  {\path{doi:10.1109/18.796389}}.

\bibitem{HsLG11}
Min-Hsiu Hsieh and Fran{\c{c}}ois Le~Gall.
\newblock {NP}-hardness of decoding quantum error-correction codes.
\newblock {\em Physical Review A}, 83(5):052331, May 2011.
\newblock \href {https://doi.org/10.1103/PhysRevA.83.052331}
  {\path{doi:10.1103/PhysRevA.83.052331}}.

\bibitem{MS77}
Florence~J. MacWilliams and Neil~J.~A. Sloane.
\newblock {\em The Theory of Error-Correcting Codes}.
\newblock North-Holland, Amsterdam, 1977.

\bibitem{KKKS06}
Avanti Ketkar, Andreas Klappenecker, Santosh Kumar, and Pradeep~Kiran
  Sarvepalli.
\newblock Nonbinary stabilizer codes over finite fields.
\newblock {\em IEEE Transactions on Information Theory}, 52(11):4892--4914,
  November 2006.
\newblock \href {https://doi.org/10.1109/TIT.2006.883612}
  {\path{doi:10.1109/TIT.2006.883612}}.

\bibitem{Gra21}
Markus Grassl.
\newblock Algebraic quantum codes: linking quantum mechanics and discrete
  mathematics.
\newblock {\em International Journal of Computer Mathematics: Computer Systems
  Theory}, 6(4):243--250, 2021.
\newblock \href {https://doi.org/10.1080/23799927.2020.1850530}
  {\path{doi:10.1080/23799927.2020.1850530}}.

\bibitem{HuPl10}
W.~Cary Huffman and Vera Pless.
\newblock {\em Fundamentals of Error-Correcting Codes}.
\newblock Cambridge University Press, 2010.
\newblock \href {https://doi.org/10.1017/CBO9780511807077}
  {\path{doi:10.1017/CBO9780511807077}}.

\bibitem{BBFKKW06}
Anton Betten, Michael Braun, Harald Fripertinger, Adalbert Kerber, Axel
  Kohnert, and Alfred Wassermann.
\newblock {\em Error-Correcting Linear Codes}.
\newblock Springer, Berlin, 2006.
\newblock \href {https://doi.org/10.1007/3-540-31703-1}
  {\path{doi:10.1007/3-540-31703-1}}.

\bibitem{IoMe07}
Lev Ioffe and Marc M{\'e}zard.
\newblock Asymmetric quantum error-correcting codes.
\newblock {\em Physical Review A}, 75(3):03245, March 2007.
\newblock \href {https://doi.org/10.1103/PhysRevA.75.032345}
  {\path{doi:10.1103/PhysRevA.75.032345}}.

\bibitem{SKR09}
Pradeep~Kiran Sarvepalli, Andreas Klappenecker, and Martin R{\"o}tteler.
\newblock Asymmetric quantum codes: construction, bounds and performance.
\newblock {\em Proceedings of the Royal Society A}, 465(2105):1645--1672, May
  2009.
\newblock \href {https://doi.org/10.1098/rspa.2008.0439}
  {\path{doi:10.1098/rspa.2008.0439}}.

\bibitem{CRSS98}
A.~Robert Calderbank, Eric~M. Rains, P.~W. Shor, and Neil J.~A. Sloane.
\newblock Quantum error correction via codes over {$GF(4)$}.
\newblock {\em IEEE Transactions on Information Theory}, 44(4):1369--1387,
  April 1998.
\newblock \href {https://doi.org/10.1109/18.681315}
  {\path{doi:10.1109/18.681315}}.

\bibitem{LaGu13}
Giuliano~G. La~Guardia.
\newblock Asymmetric quantum codes: new codes from old.
\newblock {\em Quantum Information Processing}, 12:2771–2790, August 2013.
\newblock \href {https://doi.org/s11128-013-0562-4}
  {\path{doi:s11128-013-0562-4}}.

\bibitem{FLX06}
Keqin Feng, San Ling, and Chaoping Xing.
\newblock Asymptotic bounds on quantum codes from algebraic geometry codes.
\newblock {\em IEEE Transactions on Information Theory}, 52(3):986--991, March
  2006.
\newblock \href {https://doi.org/10.1109/TIT.2005.862086}
  {\path{doi:10.1109/TIT.2005.862086}}.

\bibitem{KnLa97}
Emanuel Knill and Raymond Laflamme.
\newblock Theory of quantum error-correcting codes.
\newblock {\em Physical Review A}, 55(2):900--911, February 1997.
\newblock \href {https://doi.org/10.1103/PhysRevA.55.900}
  {\path{doi:10.1103/PhysRevA.55.900}}.

\bibitem{DaLi22}
Reza Dastbasteh and Petr Lison{\v{e}}k.
\newblock A new family of quantum codes from duadic codes.
\newblock In {\em WCC 2022: The Twelfth International Workshop on Coding and
  Cryptography}, Rostock, March 2022.
\newblock URL:
  \url{https://www.wcc2022.uni-rostock.de/storages/uni-rostock/Tagungen/WCC2022/Papers/WCC_2022_paper_19.pdf}.

\bibitem{DaLi22b}
Reza Dastbasteh and Petr Lison{\v{e}}k.
\newblock New quantum codes from self-dual codes over $\mathbb{F}_4$.
\newblock Preprint arXiv:2211.00891 [cs.IT], November 2022.
\newblock \href {https://doi.org/10.48550/arXiv.2211.00891}
  {\path{doi:10.48550/arXiv.2211.00891}}.

\bibitem{Grassl:codetables}
Markus Grassl.
\newblock {Bounds on the minimum distance of linear codes and quantum codes}.
\newblock Online available at \url{http://www.codetables.de}, 2007.
\newblock Accessed on 2022-08-10.

\bibitem{GrWh04}
Markus Grassl and Greg White.
\newblock New good linear codes by special puncturings.
\newblock In {\em Proceedings 2004 International Symposium on Information
  Theory (ISIT 2004)}, page 454, 2004.
\newblock \href {https://doi.org/10.1109/ISIT.2004.1365491}
  {\path{doi:10.1109/ISIT.2004.1365491}}.

\bibitem{Magma}
Wieb Bosma, John~J. Cannon, and Catherine Playoust.
\newblock {The Magma Algebra System I: The User Language}.
\newblock {\em Journal of Symbolic Computation}, 24(3-4):235--265, 1997.
\newblock \href {https://doi.org/10.1006/jsco.1996.0125}
  {\path{doi:10.1006/jsco.1996.0125}}.

\bibitem{FeMa04}
Keqin Feng and Zhi Ma.
\newblock A finite {G}ilbert-{V}arshamov bound for pure stabilizer quantum
  codes.
\newblock {\em IEEE Transactions on Information Theory}, 50(12):3323--3325,
  December 2004.
\newblock \href {https://doi.org/10.1109/TIT.2004.838088}
  {\path{doi:10.1109/TIT.2004.838088}}.

\end{thebibliography}
\end{document}